\documentclass{article}
\usepackage[utf8]{inputenc}

\usepackage{geometry}
\usepackage{amsmath}
\usepackage{amsfonts}
\usepackage{float}
\usepackage{stmaryrd}
\usepackage{graphicx}
\usepackage{caption}
\usepackage{subcaption}
\usepackage{authblk}
\usepackage[all]{xy}
\usepackage{tikz-cd}
\usepackage{hyperref}
\usepackage{cleveref}
\usepackage{multirow}
\usepackage{amsthm}
\usepackage{grffile}

\newcommand{\uu}{\boldsymbol{u}}
\newcommand{\mm}{\boldsymbol{m}}
\newcommand{\vv}{\boldsymbol{v}}
\newcommand{\R}{\mathbb{R}}

\newcommand{\xx}{\boldsymbol{x}}
\newcommand{\g}{\mathfrak{g}}
\newcommand{\ssig}{\boldsymbol{\Sigma}}

\newtheorem{theorem}{Theorem}
\newtheorem{proposition}[theorem]{Proposition}%

\newtheorem{problem}{Problem}%

\newtheorem{remark}{Remark}%

\newtheorem{definition}{Definition}
\newtheorem{assumption}{Assumption}

\crefname{problem}{problem}{problems}
\crefname{assumption}{assumption}{assumptions}
\title{Metriplectic formulations of variational thermodynamics}
\author{Valentin Carlier}
\date{\today}

\begin{document}

\maketitle

\begin{abstract}
We propose a metriplectic reformulation of Lagrangian variational formulations for non-equilibrium thermodynamics. 
We prove that solutions to these constrained variational principles can be generated by the sum of a classic Poisson bracket and a metriplectic 4-bracket, that takes the Hamiltonian and the entropy as generators.
We study different cases: simple systems, discrete systems, Euler-Poincaré reduced systems and systems with no symplectic part. 
Several example are shown, including infinite dimensional problems arising from continuum mechanics.
\end{abstract}

\section{Introduction}
Several formalisms have been proposed in the past decades for incorporating dissipative effects in a priori non-dissipative frameworks, such as Lagrangian or Hamiltonian mechanics. 
Many proceed by adding terms to the Poisson bracket, more precisely, considering a second bracket which is held responsible for dissipative effects. 
First attempts were done in \cite{kaufman1984dissipative} which introduces the concept of dissipative bracket, 
\cite{morrison1984bracket} which provides an interpretation of the dissipation as a bracket with an entropy variable and
\cite{morrison1986paradigm} where this second bracket is interpreted as a metric term and introduces the so called "metriplectic systems".
Another formalism called "GENERIC", presented in \cite{grmela1997dynamics}, provides a dissipative bracket accounting for thermodynamic effect.
Recently \cite{morrison2024inclusive}, proposes a unifying framework in which the dissipative effects are described by means of a 4-bracket, that has the same symmetries as a Riemann curvature tensor, leading to interesting geometric interpretations.

On the side of Lagrangian description, attempts were also made to include dissipation mechanisms in variational principles, for example the principle of least dissipation of energy \cite{onsager1931reciprocal}, or the principle of minimum entropy production \cite{prigogine1947etude}. We refer to \cite{gyarmati1970non} for a review of these variational principles. 
We also point to \cite{biot1975virtual} where a principle of virtual dissipation is introduced, that includes several irreversible thermodynamical systems.
Here we are interested with the recent developments done in \cite{gay2017lagrangian1,gay2017lagrangian2}, that present a generalized Lagrange-d'Alembert principle for non equilibrium thermodynamics. 

For non-dissipative systems, it is well known that one can go from a Lagrangian to a Hamiltonian perspective, that is, from a least action principle to a Poisson bracket-Hamiltonian pair, by using a Legendre transform on the Lagrangian to get the Hamiltonian and taking a canonical bracket (see for example \cite{marsden2013introduction}). 
In this article we study the link between the constrained (Lagrangian) variational principle of thermodynamics and the metriplectic framework of dissipative mechanics.
Understanding this link might enable the derivation of more complex models and gain more understanding on the latter, as using a generalized Lagrange-d'Alembert principle is straightforward after identifying an entropy of the system, while metriplectic brackets might be harder to derive but useful for studying the stability and other properties of a system.
Such reformulation might also be interesting in a numerical perspective, as a lot of research has be done in using symplectic/metriplectic \cite{sanz1992symplectic,kraus2017metriplectic} or variational \cite{pavlov2011structure,gawlik2011geometric,gawlik2021variational,carlier2025variational} structure to create structure preserving algorithms. Hence understanding the bridge between formalisms might help in devising new methods.

To do so, we will proceed in a similar way as for non-dissipative systems: given a Lagrangian, we will use a Legendre transform to introduce the canonical momentum and Hamiltonian.
We stress that our formalism is only valid when the Legendre transform is invertible (except for \cref{sec:no_symp}), as is the standard Lagrangian to Hamiltonian transformation.
Therefore we are not able to describe systems such as the ones presented in \cite{eldred2021thermodynamically}.
We will then study the evolution of an observable of the system in the aim to rewrite it using a canonical bracket and a dissipative bracket. The bracket obtained will have a simple structure under the assumption that the friction force is given by a symmetric tensor contracted with the generalized velocity of the system. 
This procedure will be conduce for simple, discrete and Euler-Poincarré type of systems, giving a general metriplectic formulation for dissipative systems built using the generalized Lagrange-d'Alembert principle.
Several examples coming from \cite{gay2017lagrangian1,gay2017lagrangian2} will be presented to illustrate this transformation, in particular we present a metriplectic formulation for simple and multiple pistons problems, as well as for the thermodynamical Navier-Stokes equation and systems with chemical reactions. 

The remainder of this work is organized as follows: in \cref{sec:simple} we present our results for simple systems, then \cref{sec:discrete} shows how to apply the same method for discrete systems. \Cref{sec:EP} presents similar results in the context of reduced Euler-Poincaré systems, while \cref{sec:no_symp} focuses on particular systems with no symplectic part. \Cref{sec:concl} gather some concluding remarks.

\section{Simple systems}
\label{sec:simple}
Our first subject of interest is the variational formulation for simple (mechanic) systems with thermodynamics. 
Following \cite{stueckelberg1974thermocinetique,gay2017lagrangian1}, we define as simple systems, 
systems that can be described by one thermodynamic variable (that can always be chosen to be the entropy  of the system, denoted by $S$) and a finite number of mechanical variables encompassed in a finite dimensional variable $q \in Q$.   
We consider a phase space $Q \subset \R^d$ for simplicity and our system is described by a Lagrangian $L(q, \dot{q},S)$. 
We suppose that the system is subject to a friction force $F^{fr}(q, \dot{q}, S)$, which is responsible for the creation of entropy. 
We will first recall the variational formulation associated with such a system, before exhibiting a metriplectic 4-bracket which gives the same dynamic (considering canonical Hamiltonian and Poisson bracket for the symplectic part).
\subsection{Variational formulation}
The variational formulation for simple thermodynamic systems is given by: 
\begin{problem}[Variational formulation for simple thermodynamic system]
\label{pbm:variational_form_simple}
Find a curve $(q(t), S(t)) \in Q \times \R$ that satisfies the variational condition
\begin{equation}
\label{eqn:lagrangian_simple}
\delta \int_0^{T_f} L(q, \dot{q}, S) dt = 0 ~,
\end{equation}
with variations subject to the following variational constraint
\begin{equation}
\label{eqn:var_constraint_simple}
\frac{\partial L}{\partial S} \delta S = \langle F^{fr}(q,\dot{q},S), \delta q \rangle ~,
\end{equation}
and the phenomenological constraint
\begin{equation}
\label{eqn:pheno_constraint_simple}
\frac{\partial L}{\partial S} \dot{S} = \langle F^{fr}(q,\dot{q},S), \dot{q} \rangle ~.
\end{equation}
\end{problem}
Such a constraint is a priori non-holonomic as in all generality we can not guarantee the existence of a submanifold $N$ of $Q \times \R$ such that this constraint is equivalent of being restricted to $N$.
The duality product $\langle \cdot , \cdot \rangle$ is the standard dot product of $\R^d$.
The solution curves to this problem, obtained in the standard way by computing the variation and integrating by part the term in $\dot{\delta q}$ (see \cite{gay2017lagrangian1} for detailed computations), satisfy: 
\begin{subequations}
\label{eqn:EL_simple}
\begin{equation}
\label{eqn:EL_simple_dot_q}
\frac{d}{dt} \frac{\partial L}{\partial \dot{q}} - \frac{\partial L}{\partial q} = F^{fr}(q, \dot{q}, S) ~,
\end{equation}
\begin{equation}
\label{eqn:EL_simple_dot_s}
\frac{\partial L}{\partial S} \dot{S} = \langle F^{fr}(q,\dot{q},S), \dot{q} \rangle ~.
\end{equation}
\end{subequations}
The first equation is similar to the standard Euler-Lagrange equations given by critical curves of a Lagrangian, adding the dissipative force (which is not a force derivated from a potential) to this "momentum" equation. The second equation is \cref{eqn:pheno_constraint_simple} and encodes entropy dissipation.
\subsection{A first metriplectic formulation}
\label{sec:1st_metri_simple}
Let us first introduce the canonical momentum and Hamiltonian, in the same way as in the standard non-dissipative set-up \cite{marsden2013introduction}: 
\begin{definition}
\label{def:mom_Ham_simple}
Given a Lagrangian $L(q, \dot{q}, S)$, the canonical momentum is defined as 
\begin{equation}
\label{eqn:mom_simple}
p = \frac{\partial L}{\partial \dot{q}} ~,
\end{equation}
and the Hamiltonian (or energy) of the system as 
\begin{equation}
\label{eqn:ham_simple}
H(q,p,S) = \langle p, \dot{q} \rangle - L ~.
\end{equation}
\end{definition}
For the remainder of this work, except in \cref{sec:no_symp} we will assumer that the Legendre transform $(q, \dot{q}) \mapsto (q, p)$ is a diffeomorphism.
\begin{proposition}
The equations of motion in the $(q,p,S)$ variables read: 
\begin{subequations}
\label{eqn:motion_ham_simple}
\begin{equation}
\dot{q} = \frac{\partial H}{\partial p} ~,
\end{equation}
\begin{equation}
\dot{p} = -\frac{\partial H}{\partial q} + F^{fr}(q,\dot{q},S)~,
\end{equation}
\begin{equation}
\frac{\partial H}{\partial S}\dot{S} = - \langle F^{fr}(q,\dot{q},S), \dot{q} \rangle ~.
\end{equation}
\end{subequations}
\end{proposition}
\begin{proof}
This is a simple rewriting of \cref{eqn:EL_simple} using the change of variables $(q, \dot{q}, S) \rightarrow (q, p, S)$ 
\end{proof}
We remark that the first equation is standard in Hamiltonian dynamics, the second one consists of a standard term plus the friction force (this is a second Newton law, with a potential force in $\frac{\partial H}{\partial q}$ and a dissipative force $F^{fr}(q,\dot{q},S)$). The final equation is again the entropy dissipation. Written in this form the latter can be seen as an equation added to restore energy (Hamiltonian) conservation.

Our aim is to write the evolution of any function of the dynamic as a metriplectic system:

\begin{definition}[Metriplectic system]
\label{def:metriplectic}
Consider a symplectic 2-bracket $\{ \cdot , \cdot \}$ (that is antisymmetric, bilinear and satisfying the Jacobi identity),
and a metric 4-bracket $(\cdot,\cdot;\cdot,\cdot)$ with the following identity $(F,G;M,N) = -(G,F;M,N) = -(F,G;N,M) = (M,N;F,G)$ \cite{morrison2024inclusive}, both taking as argument functions of the dynamic.
Consider also two function of the dynamic: the Hamiltonian $H$ and the entropy $S$, that have the property that $\{F,S\}=0$ for all $F$.
The dynamic is a metriplectic system if for all function of the system $F$, we have $\dot{F} = \{F,H\} + (F,H;S,H)$
\end{definition}

\begin{remark}
If the dynamic satisfies $\dot{F} = \{F,H\}$ we have a standard Hamiltonian system, and the entropy is a Casimir of the system. The metric 4-bracket is responsible for the dissipation.
\end{remark}

Let consider a function of the dynamic $F(q(t),p(t),S(t))$ and we conduct the following computation, where we denote $T = \frac{\partial H}{\partial S}$ and $K = -\langle F^{fr}, \frac{\partial H}{\partial p}\rangle$ the temperature of the system and the work of the dissipative force: 
\begin{align*}
\dot{F} &= \langle \frac{\partial F}{\partial q} , \dot{q} \rangle + \langle \frac{\partial F}{\partial p}, \dot{p} \rangle + \frac{\partial F}{\partial S} \dot{S} \\
	&= \langle \frac{\partial F}{\partial q}, \frac{\partial H}{\partial p} \rangle - \langle \frac{\partial F}{\partial p}, \frac{\partial H}{\partial q} \rangle + \langle F^{fr}, \frac{\partial F}{\partial p}\rangle - \frac{\partial F}{\partial S} \frac{\langle F^{fr}, \frac{\partial H}{\partial p} \rangle}{T} \\
	&= \{F,H\} + \frac{1}{T}\left(\langle F^{fr}, \frac{\partial F}{\partial p}\rangle \frac{\partial H}{\partial S} - \langle F^{fr}, \frac{\partial H}{\partial p} \rangle \frac{\partial F}{\partial S}\right) \\
	&= \{F,H\} + \frac{1}{T K}\left(\langle F^{fr}, \frac{\partial F}{\partial p}\rangle \frac{\partial H}{\partial S} - \langle F^{fr}, \frac{\partial H}{\partial p} \rangle \frac{\partial F}{\partial S}\right)\left(\langle F^{fr}, \frac{\partial S}{\partial p}\rangle \frac{\partial H}{\partial S} - \langle F^{fr}, \frac{\partial H}{\partial p} \rangle \frac{\partial S}{\partial S}\right) \\
	&= \{F,H\} + (F,H;S,H) ~,
\end{align*}
with
\begin{subequations}
\label{eqn:metriplectic_brackets_simple}
\begin{equation}
\label{eqn:Poisson-bracket}
\{F,G\} = \langle \frac{\partial F}{\partial q}, \frac{\partial G}{\partial p} \rangle - \langle \frac{\partial F}{\partial p}, \frac{\partial G}{\partial q} \rangle ~,
\end{equation}
\begin{equation}
\label{eqn:metric-bracket}
(F,G;M,N) = \frac{1}{T K}\left(\langle F^{fr}, \frac{\partial F}{\partial p}\rangle \frac{\partial G}{\partial S} - \langle F^{fr}, \frac{\partial G}{\partial p} \rangle \frac{\partial F}{\partial S}\right)\left(\langle F^{fr}, \frac{\partial M}{\partial p}\rangle \frac{\partial N}{\partial S} - \langle F^{fr}, \frac{\partial N}{\partial p} \rangle \frac{\partial M}{\partial S}\right) ~.
\end{equation}
\end{subequations}
Where we have used that $\frac{\partial S}{\partial p} = 0$ and $\frac{\partial S}{\partial S} = 1$. The bracket given by \cref{eqn:metric-bracket} satisfies all the symmetry conditions of a metriplectic 4-bracket. We reformulate our result in the following theorem:
\begin{theorem}
A curve $(q(t),S(t)) \in Q \times \R$ is solution to \cref{pbm:variational_form_simple} if and only if, doing the change of variables $(q, \dot{q}, S) \mapsto (q, p, S)$, it is a metriplectic system \ref{def:metriplectic} with the brackets defined by \cref{eqn:metriplectic_brackets_simple} and the Hamiltonian as in \cref{eqn:ham_simple} and entropy $S$.
\end{theorem}
\begin{remark}
If $F^{fr}$ is really dissipative, then $K>0$ and we have entropy dissipation at rate $\frac{K}{T}$.
\end{remark}
\subsection{A second metriplectic formulation}
\label{sec:2nd_metri_simple}
The formulation given in the previous section has one major drawback: the bracket is ill defined if $K=0$. This situation could arise if one uses this framework to model a non-conservative force that is orthogonal to the motion, or if at some point the system stop moving (for example around equilibrium points). We here provide another formulation that overcomes this flaw, but rely on an assumption on the structure of the dissipative force. This second formulation has also the advantage of providing a bracket that has a Kulkarni-Nomizu product structure and is therefore more susceptible to have a geometric interpretation.

\begin{assumption}
\label{assump:structure_dissipation}
For the remainder of this work we assume that dissipation forces can be written as
\begin{equation}
\label{eqn:structure_dissipation}
F^{fr} = - \Lambda(q,\dot{q},S) \dot{q} ~, ~ \Lambda \text{ symmetric matrix.}
\end{equation}
\end{assumption}

This hypothesis is not very restrictive: such a symmetric matrix will always exist as long as $F^{fr}$ is zero when $\dot{q}$ is. It has the physical meaning of assuming that the friction force cancels out when the system is not in motion. Moreover in a lot of physically relevant systems, the dissipation is most commonly given in this form (is the case of all the examples presented in \cite{gay2017lagrangian1,gay2017lagrangian2}).

We consider a functional $F$ and resume our calculation from the previous subsection, using our additional \cref{assump:structure_dissipation} (remind that $T=\frac{\partial H}{\partial S}$):
\begin{align*}
\dot{F} &= \langle \frac{\partial F}{\partial q} , \dot{q} \rangle + \langle \frac{\partial F}{\partial p}, \dot{p} \rangle + \frac{\partial F}{\partial S} \dot{S} \\
	&= \langle \frac{\partial F}{\partial q}, \frac{\partial H}{\partial p} \rangle - \langle \frac{\partial F}{\partial p}, \frac{\partial H}{\partial q} \rangle + \langle F^{fr}, \frac{\partial F}{\partial p}\rangle - \frac{\partial F}{\partial S} \frac{\langle F^{fr}, \frac{\partial H}{\partial p} \rangle}{T} \\
	&= \{F,H\} + \frac{1}{T}\left(\langle \frac{\partial H}{\partial p},\Lambda \frac{\partial H}{\partial p} \rangle \frac{\partial F}{\partial S} - \langle \frac{\partial H}{\partial p},\Lambda \frac{\partial F}{\partial p}\rangle \frac{\partial H}{\partial S}\right) \\
	&= \{F,H\} + \frac{1}{T}\left(
	\begin{array}{cr}
	\langle \frac{\partial H}{\partial p},\Lambda \frac{\partial H}{\partial p} \rangle \frac{\partial F}{\partial S}\frac{\partial S}{\partial S} 
	-\langle \frac{\partial H}{\partial p},\Lambda \frac{\partial F}{\partial p}\rangle \frac{\partial H}{\partial S}\frac{\partial S}{\partial S}
	+\langle \frac{\partial S}{\partial p},\Lambda \frac{\partial F}{\partial p}\rangle \frac{\partial H}{\partial S}\frac{\partial H}{\partial S} 
	- \langle \frac{\partial S}{\partial p},\Lambda \frac{\partial H}{\partial p} \rangle \frac{\partial F}{\partial S}\frac{\partial H}{\partial S}
	\end{array}
	\right) \\
	&= \{F,H\} + (F,H;S,H)
\end{align*}
With now
\begin{equation}
\label{eqn:metric_bracket_sym}
(F,G;M,N) = \frac{1}{T} \left(
\begin{array}{cr}
 \langle \frac{\partial N}{\partial p},\Lambda \frac{\partial G}{\partial p} \rangle \frac{\partial F}{\partial S}\frac{\partial M}{\partial S} 
 -\langle \frac{\partial N}{\partial p},\Lambda \frac{\partial F}{\partial p}\rangle \frac{\partial G}{\partial S}\frac{\partial M}{\partial S}
 +\langle \frac{\partial M}{\partial p},\Lambda \frac{\partial F}{\partial p}\rangle \frac{\partial G}{\partial S}\frac{\partial N}{\partial S} 
 - \langle \frac{\partial M}{\partial p},\Lambda \frac{\partial G}{\partial p} \rangle \frac{\partial F}{\partial S}\frac{\partial N}{\partial S}
 \end{array}
 \right)
\end{equation}
We remark that this 4-bracket has a Kulkarni-Nomizu product \cite{kulkarni1972bianchi,morrison2024inclusive} structure between the 2-covariant tensors given by $d_p f \frac{\Lambda}{T} d_p g$ and $d_s f \cdot d_s g$.
The Kulkarni-Nomizu product of 2 symmetric 2-bracket $a(\cdot, \cdot)$ and $b(\cdot, \cdot)$ is given by 
\begin{equation}
\label{eqn:KNproduct}
a \owedge b(F,G;M,N) = a(F,M)b(G,N) - a(F,N)b(G,M) + a(G,N)b(F,M) - a(G,M)b(F,N) ~.
\end{equation}
We recover the dissipation rate for entropy $ \frac{1}{T}\langle \frac{\partial H}{\partial p},\Lambda \frac{\partial H}{\partial p} \rangle$, which is positive if we impose the condition that $\Lambda$ is definite positive. Let us reformulate this result in a theorem: 
\begin{theorem}
A curve $(q(t),S(t)) \in Q \times \R$ is solution to \cref{pbm:variational_form_simple} if and only if, doing the change of variables $(q, \dot{q}, S) \mapsto (q, p, S)$, it is a metriplectic system \ref{def:metriplectic} with the brackets defined by \cref{eqn:Poisson-bracket,eqn:metric_bracket_sym}  and the Hamiltonian as in \cref{eqn:ham_simple}.
\end{theorem}
\begin{remark}
Those two variational formulation are valid under the assumption that the Legendre transform $(q, \dot{q}) \mapsto (q, p)$ is well defined and a diffeomorphism. We are therefore not able to treat systems such as chemical reaction as defined in sec 3.3 of \cite{gay2017lagrangian1} where the Lagrangian is independent of $\dot{q}$, for a treatment of such systems, see \cref{sec:no_symp}. However our construction is still valid for all the problems where the Lagrangian has the particular form $L = K-U$ where $K$ is the kinetic energy and $U$ the potential energy as most of the systems coming from the study of fluids and plasmas. 
\end{remark}

We can compare our result with the previously obtained dissipative brackets. In \cite{gay2020variational} the authors give a reformulation of variational thermodynamics using single and double generator as well as 2-metriplectic bracket.
The 2-metriplectic bracket is obtained directly from the 4-metriplectic bracket as $(F,G) = (F,H;G,H)$.

In \cite{gay2020variational}, the authors had already identified the 2-bracket 
\begin{equation}
\label{eqn:metric_2_sym}
(F,G) = \frac{1}{T} \left(
\begin{array}{cr}
 \langle \frac{\partial H}{\partial p},\Lambda \frac{\partial H}{\partial p} \rangle \frac{\partial F}{\partial S}\frac{\partial G}{\partial S} 
 -\langle \frac{\partial H}{\partial p},\Lambda \frac{\partial F}{\partial p}\rangle \frac{\partial H}{\partial S}\frac{\partial G}{\partial S}
 +\langle \frac{\partial G}{\partial p},\Lambda \frac{\partial F}{\partial p}\rangle \frac{\partial H}{\partial S}\frac{\partial H}{\partial S} 
 - \langle \frac{\partial G}{\partial p},\Lambda \frac{\partial H}{\partial p} \rangle \frac{\partial F}{\partial S}\frac{\partial H}{\partial S}
 \end{array} ~,
 \right)
\end{equation}
which matches the reduction of the 4-bracket we derive here to a standard 2-bracket.

\subsection{Example: one piston problem}
We consider a system consisting of a piston in a cylinder which is occupied by a perfect gas. 
The system is described by $x$ the position of the piston ($x=0$ meaning the piston is in contact with the bottom of the cylinder) and $S$, the entropy of the gas in the cylinder, assuming the homogeneity inside the cylinder. 
We refer to \cite{gay2017lagrangian1} for a more detailed description of the problem. The Lagrangian is given by 
\begin{equation}
\label{eqn:Lagrangian_piston}
L(x, \dot{x}, S) = \frac{1}{2}m \dot{x}^2 - U(x,S) ~,
\end{equation}
where $m$ is the mass of the piston and $U$ its internal energy. The friction force is explicitly given by 
\begin{equation}
\label{eqn:Dissip_piston}
F^{fr} = - \lambda(x,S) \dot{x} \text{ with } \lambda \geq 0 ~,
\end{equation}
which has the form of \cref{eqn:structure_dissipation}. Following our computations from \cref{sec:1st_metri_simple,sec:2nd_metri_simple}, we define the momentum, Hamiltonian and brackets
\begin{subequations}
\label{eqn:mom_Ham_brack_piston}
\begin{equation}
p = m\dot{x} ~,
\end{equation}
\begin{equation}
H = \frac{p^2}{2m} + U(x,S) ~,
\end{equation}
\begin{equation}
\{F,G\} = \langle \frac{\partial F}{\partial x}, \frac{\partial G}{\partial p} \rangle - \langle \frac{\partial F}{\partial p}, \frac{\partial G}{\partial x} \rangle ~,
\end{equation}
\begin{equation}
(F,G;M,N) = \frac{\lambda}{T}\left(\frac{\partial N}{\partial p} \frac{\partial G}{\partial p} \frac{\partial F}{\partial S}\frac{\partial M}{\partial S} 
	-\frac{\partial N}{\partial p} \frac{\partial F}{\partial p} \frac{\partial G}{\partial S}\frac{\partial M}{\partial S}
	+ \frac{\partial M}{\partial p} \frac{\partial F}{\partial p} \frac{\partial G}{\partial S}\frac{\partial N}{\partial S} 
	- \frac{\partial M}{\partial p} \frac{\partial G}{\partial p} \frac{\partial F}{\partial S}\frac{\partial N}{\partial S}
	\right) ~.
\end{equation}
\end{subequations}
And the dynamic of the system is described by $\dot{F}= \{F,H\} + (F,H;S,H)$. Inserting for $F$ the functions $p$ and $S$ we find 
\begin{equation}
\label{eqn:dynamic_piston}
\dot{p}=-\frac{\partial U}{\partial x} -\lambda \frac{p}{m} \qquad \dot{S} = \frac{\lambda}{T} \frac{p^2}{m^2}
\end{equation}
which correspond to the equation of the piston as given in \cite{gruber1999thermodynamics}
\section{Discrete systems}
\label{sec:discrete}
We now turn to discrete systems, that is systems composed of $N$ simple systems interacting. Such a system can be described by the set of all the mechanical variables in a phase space $Q$ (we keep $Q \subset \R^d$ for simplicity) and a set of entropies $(S_i)_{i=1,...,N}$ (one per system). The evolution of the system is described by a Lagrangian $L(q, \dot{q}, (S_i))$. Every system has a friction force $F^{fr(i)}$ and heat exchanges within the systems are modelled via coefficients $\kappa_{ij} \geq 0$, symmetric, such that the power of the heat transfer from system $j$ to system $i$ is $\kappa_{ij}(T_j-T_i)$, where $T_i=\frac{\partial H}{\partial S_i}$.
\subsection{Variational principle}
The variational formulation for a discrete thermodynamical system is given by \cite{gay2017lagrangian1} 
\begin{problem}[Variational formulation for discrete thermodynamic systems]
\label{pbm:variational_form_discrete}
Find a curve $(q(t), S_i(t), \Sigma_i(t), \Gamma_i(t)) \in Q \times \R^N \times \R^N \times \R^N$ that satisfies the variational condition
\begin{equation}
\label{eqn:lagrangian_discrete}
\delta \int_0^{T_f} L(q, \dot{q}, S_1,...,S_N) + \sum\limits_{i=1}^N (S_i-\Sigma_i)\dot{\Gamma}^i dt = 0 ~,
\end{equation}
with variations subject to the following variational constraints
\begin{equation}
\label{eqn:var_constraint_discrete}
\frac{\partial L}{\partial S_i} \delta \Sigma_i = \langle F^{fr(i)}(...), \delta q \rangle + \sum\limits_{j=1}^N J_j^i(...) \delta \Gamma^j~,
\end{equation}
and phenomenological constraints
\begin{equation}
\label{eqn:pheno_constraint_discrete}
\frac{\partial L}{\partial S_i} \dot{\Sigma_i} = \langle F^{fr(i)}(...), \dot{q} \rangle + \sum\limits_{j=1}^N J_j^i(...) \dot{\Gamma^j}~.
\end{equation}
\end{problem}
Where $J_j^i$ is the "friction force" associated with $\Sigma^i$ and is given by $J_j^i = - \kappa_{ij} + \delta_{ij} \sum\limits_{k=1}^N \kappa_{ik}$. 
The additional variables $\Gamma_i$ are interpreted as thermal displacement, while the $\Sigma_i$ correspond to the total entropy production (which in this case is equal to the entropy production).
The solution curves to this problem are given by 
\begin{subequations}
\label{eqn:EL_discrete}
\begin{equation}
\label{eqn:EL_discrete_dot_q}
\frac{d}{dt} \frac{\partial L}{\partial \dot{q}} - \frac{\partial L}{\partial q} = \sum\limits_{i=1}^NF^{fr(i)} ~,
\end{equation}
\begin{equation}
\label{eqn:EL_discrete_dot_sigma}
\frac{\partial L}{\partial S_i} \dot{\Sigma_i} = \langle F^{fr(i)}, \dot{q} \rangle + \sum\limits_{j=1}^N J_j^i \dot{\Gamma_j} ~,
\end{equation}
\begin{equation}
\label{eqn:EL_discrete_dot_s}
\dot{S_i}=\dot{\Sigma_i} ~,
\end{equation}
\begin{equation}
\label{eqn:EL_discrete_dot_gamma}
\dot{\Gamma_i}=- \frac{\partial L}{\partial S_i} ~.
\end{equation}
\end{subequations}
The first equation is very similar to the one obtained for simple systems \cref{eqn:EL_simple_dot_q}, summing the contributions of dissipative forces, as the momentum equation is not affected by the heat transfer between the different systems. Substituting \cref{eqn:EL_discrete_dot_s,eqn:EL_discrete_dot_gamma} in \cref{eqn:EL_discrete_dot_sigma} we find the evolution equation for entropies
\begin{equation}
\label{eqn:EL_discrete_entropy}
\frac{\partial L}{\partial S_i} \dot{S_i} = \langle F^{fr(i)}, \dot{q} \rangle - \sum\limits_{j=1}^N J_j^i \frac{\partial L}{\partial S_j} ~,
\end{equation}
they are similar to the one for simple systems \cref{eqn:EL_simple_dot_s}, with a new term accounting for heat exchange between the different systems.
\subsection{Metriplectic reformulation}
We start by expressing the equations of motion in term of momentum and Hamiltonian ;
\begin{proposition}
\label{prop:motion_ham_discrete}
Defining $H$ and $p$ as in \cref{def:mom_Ham_simple} the equations of motion in the $(q,p,(S_i))$ variables are: 
\begin{subequations}
\begin{equation}
\dot{q} = \frac{\partial H}{\partial p} ~,
\end{equation}
\begin{equation}
\dot{p} = -\frac{\partial H}{\partial q} + \sum\limits_{i=1}^N F^{fr(i)}~,
\end{equation}
\begin{equation}
\frac{\partial H}{\partial S_i}\dot{S_i} = - \langle F^{fr(i)}(q,\dot{q},S), \dot{q} \rangle - \sum\limits_{j=1}^N J_j^i \frac{\partial H}{\partial S_j}~.
\end{equation}
\end{subequations}
\end{proposition}
\begin{proof}
This is a simple rewriting of \cref{eqn:EL_simple} using the change of variables $(q, \dot{q}, (S_i)) \rightarrow (q, p, (S_i))$ 
\end{proof}
As before, we want to write the evolution of any function $F(q(t),p(t),(S_i(t)))$ as $\dot{F} = \{F,H\} + (F,H;S,H)$ where $S=\sum_i S_i$ is the total entropy. To do so we will conduct a similar computation, introducing $T_i = \frac{\partial H}{\partial S_i}$, the temperature of system $i$, and using the assumption \ref{assump:structure_dissipation} to write $F^{fr(i)} = - \Lambda_i \dot{q}$ 
\begin{align*}
\dot{F} &= \langle \frac{\partial F}{\partial q}, \dot{q} \rangle + \langle \frac{\partial F}{\partial p}, \dot{p} \rangle + \sum\limits_{i=1}^N \frac{\partial F}{\partial S_i} \dot{S_i} \\
	&= \langle \frac{\partial F}{\partial q}, \frac{\partial H}{\partial p} \rangle - \langle \frac{\partial F}{\partial p}, \frac{\partial H}{\partial q} \rangle + \sum\limits_{i=1}^N \langle F^{fr(i)}, \frac{\partial F}{\partial p}\rangle - \sum\limits_{i=1}^N \frac{\partial F}{\partial S_i} \frac{\langle F^{fr(i)}, \frac{\partial H}{\partial p} \rangle}{T_i} - \sum\limits_{i,j=1}^N \frac{\partial F}{\partial S_i} J_j^i \frac{T_j}{T_i} \\
	&= \{F,H\} + \sum\limits_{i=1}^N (F,H;S,H)_{fr(i)} - \sum\limits_{i,j=1}^N \frac{\partial F}{\partial S_i} J_j^i \frac{T_j}{T_i}
\end{align*}
with
\begin{subequations}
\label{eqn:metriplectic_brackets_discrete}
\begin{equation}
\{F,G\} = \langle \frac{\partial F}{\partial q}, \frac{\partial H}{\partial p} \rangle - \langle \frac{\partial F}{\partial p}, \frac{\partial H}{\partial q} \rangle ~,
\end{equation}
\begin{equation}
(F,G;M,N)_{fr(i)} = \frac{1}{T_i} \left( 
\begin{array}{cr}
	\langle \frac{\partial N}{\partial p},\Lambda_i \frac{\partial G}{\partial p} \rangle \frac{\partial F}{\partial S_i}\frac{\partial M}{\partial S_i} 
	-\langle \frac{\partial N}{\partial p},\Lambda_i \frac{\partial F}{\partial p}\rangle \frac{\partial G}{\partial S_i}\frac{\partial M}{\partial S_i}
	+\langle \frac{\partial M}{\partial p},\Lambda_i \frac{\partial F}{\partial p}\rangle \frac{\partial G}{\partial S_i}\frac{\partial N}{\partial S_i} 
	- \langle \frac{\partial M}{\partial p},\Lambda_i \frac{\partial G}{\partial p} \rangle \frac{\partial F}{\partial S_i}\frac{\partial N}{\partial S_i}
\end{array}
 \right) ~.
\end{equation}
\end{subequations}
Where we have used that $\frac{\partial S}{\partial p} = 0$ and $\frac{\partial S}{\partial S_i} = 1$ as before. We now need to treat the last term
\begin{align*}
- \sum\limits_{i,j=1}^N \frac{\partial F}{\partial S_i} J_j^i \frac{T_j}{T_i} &= \sum\limits_{i,j=1}^N \kappa_{ij}  \frac{\partial F}{\partial S_i} \frac{T_j}{T_i} - \sum\limits_{i,j=1}^N \delta_{ij} \sum\limits_{k=1}^N \kappa_{ik} \frac{\partial F}{\partial S_i} \frac{T_j}{T_i} \\
	&= \sum\limits_{i,j=1}^N \kappa_{ij}  \frac{\partial F}{\partial S_i} \frac{T_j}{T_i} - \sum\limits_{i,j=1}^N \kappa_{ij} \frac{\partial F}{\partial S_i} \\
	&= \sum\limits_{i,j=1}^N \kappa_{ij}  \frac{\partial F}{\partial S_i} (\frac{T_j}{T_i}-1) \\
    &= \sum\limits_{i,j=1}^N \frac{\kappa_{ij}}{2} \left(\frac{\partial F}{\partial S_i} (\frac{T_j}{T_i}-1) + \frac{\partial F}{\partial S_j} (\frac{T_i}{T_j}-1)\right) \\
    &= \sum\limits_{i,j=1}^N \frac{\kappa_{ij}}{2} (T_j-T_i) \left(\frac{\partial F}{\partial S_i} \frac{1}{T_i} - \frac{\partial F}{\partial S_j} \frac{1}{T_j}\right) \\
    &= \sum\limits_{i,j=1}^N \frac{\kappa_{ij}}{2} \frac{1}{T_i T_j} \left(\frac{\partial S}{\partial S_i}\frac{\partial H}{\partial S_j}-\frac{\partial S}{\partial S_j}\frac{\partial H}{\partial S_i})(\frac{\partial F}{\partial S_i} \frac{\partial H}{\partial S_j} - \frac{\partial F}{\partial S_j} \frac{\partial H}{\partial S_i} \right) \\,
\end{align*}
using the symmetry of the coefficients $\kappa$. If we define 
\begin{subequations}
\begin{equation}
\label{eqn:discrete_transfer_bracket}
(F,G;M,N)_{tr} = \sum\limits_{i,j=1}^N \frac{\kappa_{ij}}{2} \frac{1}{T_i T_j} \left(\frac{\partial F}{\partial S_i} \frac{\partial G}{\partial S_j} - \frac{\partial F}{\partial S_j} \frac{\partial G}{\partial S_i})(\frac{\partial M}{\partial S_i}\frac{\partial N}{\partial S_j}-\frac{\partial M}{\partial S_j}\frac{\partial N}{\partial S_i}\right) ~,
\end{equation}
\begin{equation}
(\cdot,\cdot;,\cdot,\cdot) = \sum\limits_{i=1}^N (\cdot,\cdot;,\cdot,\cdot)_{fr(i)} + (\cdot,\cdot;\cdot,\cdot)_{tr}
\end{equation}
\end{subequations}
we indeed have $\dot{F} = \{F,H\} + (F,H;S,H)$ where the brackets are satisfying the requirements.
\begin{remark}
The last bracket accounting for entropy creation through heat exchange between the systems can also be rewritten as 
\begin{align*}
(F,G;M,N)_{tr} &= \sum\limits_{i > j} \kappa_{ij} \frac{1}{T_i T_j} \left(\frac{\partial F}{\partial S_i} \frac{\partial G}{\partial S_j} - \frac{\partial F}{\partial S_j} \frac{\partial G}{\partial S_i})(\frac{\partial M}{\partial S_i}\frac{\partial N}{\partial S_j}-\frac{\partial M}{\partial S_j}\frac{\partial N}{\partial S_i}\right) \\
&= \sum\limits_{i > j} (F,G;M,N)_{tr(i,j)} ~, 
\end{align*}
where we can clearly see the contribution of the exchange between each pair of distinct systems. We again remark that this 4-brackets has a Kulkarni-Nomizu product structure between the 2-covariant tensors given by $d_{s_i} f \frac{1}{T_i} d_{s_i} g$ and $d_{s_j} f \frac{1}{T_i} d_{s_j} g$
\end{remark}
Let us summarize what we have just proved:
\begin{theorem}
A curve $(q(t),(S_i(t)),(\Sigma_i(t)),(\Gamma_i(t))) \in Q \times \R^N \times \R^N \times \R^N$ is solution to \cref{pbm:variational_form_discrete} if and only if, doing the change of variables $(q, \dot{q}, (S_i)) \mapsto (q, p, (S_i))$, it is a metriplectic system \ref{def:metriplectic} with the brackets defined by with the brackets defined by \cref{eqn:Poisson-bracket,eqn:discrete_transfer_bracket}, the Hamiltonian as in \cref{eqn:ham_simple} and $S = \sum\limits_{i=1}^N S_i$, and $\Sigma$ and $\Gamma$ satisfy \cref{eqn:EL_discrete_dot_s,eqn:EL_discrete_dot_gamma}. 
\end{theorem}
The equations satisfied by $\Sigma$ and $\Gamma$ were added for completeness, but in practice those are auxiliary variables of the variational principle that do not play any role in the dynamic.
\subsection{Example: two connected piston problem}
As illustration of two connected systems we consider the problem of pistons connected by a thermically isolated heat conducting rod. The total mass of the rod and the two piston is denoted by M, and the position of the system is characterized by $x$, the distance of the left piston to the bottom of the left cylinder. The internal energy of the left (resp. right) system is denoted by $U_1(x,S_1)$ (resp. $U_2(x,S_2)$) where $S_1$ (resp. $S_2$) denotes the entropy of the left (resp. right) system (see again \cite{gay2017lagrangian1} sec. 4.3 for more detailed description of the problem). The Lagrangian of the total system is:
\begin{equation}
\label{eqn:lagrangian_two_piston}
L(x,\dot{x},S_1,S_2) = \frac{1}{2} M \dot{x}^2 - U_1(x, S_1) - U_2(x, S_2) ~,
\end{equation} 
in each piston there is a friction force given by 
\begin{equation}
\label{eqn:friction_two_piston}
F^{fr(i)} = - \lambda_i \dot{x}
\end{equation}
which satisfies \cref{eqn:structure_dissipation}, and the heat transfer from the left system to the right is given by $\kappa(T_1-T_2)$. According to the previous theory, we define 
\begin{subequations}
\label{eqn:mom_Ham_brack_two_piston}
\begin{equation}
p = M\dot{x} ~,
\end{equation}
\begin{equation}
H = \frac{p^2}{2M} + U_1(x,S_1) + U_2(x,S_2) ~.
\end{equation}
\begin{equation}
\{F,G\} = \langle \frac{\partial F}{\partial x}, \frac{\partial G}{\partial p} \rangle - \langle \frac{\partial F}{\partial p}, \frac{\partial G}{\partial x} \rangle ~,
\end{equation}
\begin{equation}
(F,G;M,N)_{fr(i)} = \frac{\lambda_i}{T_i} \left(\frac{\partial N}{\partial p} \frac{\partial G}{\partial p} \frac{\partial F}{\partial S_i}\frac{\partial M}{\partial S_i} -\frac{\partial N}{\partial p} \frac{\partial F}{\partial p} \frac{\partial G}{\partial S_i}\frac{\partial M}{\partial S_i}
	+ \frac{\partial M}{\partial p} \frac{\partial F}{\partial p} \frac{\partial G}{\partial S_i}\frac{\partial N}{\partial S_i} - \frac{\partial M}{\partial p} \frac{\partial G}{\partial p} \frac{\partial F}{\partial S_i}\frac{\partial N}{\partial S_i} \right)
\end{equation}
\begin{equation}
(F,G;M,N)_{tr} = \frac{\kappa}{T_1 T_2} \left(\frac{\partial F}{\partial S_1} \frac{\partial G}{\partial S_2} - \frac{\partial F}{\partial S_2} \frac{\partial G}{\partial S_1})(\frac{\partial M}{\partial S_1}\frac{\partial N}{\partial S_2}-\frac{\partial M}{\partial S_2}\frac{\partial N}{\partial S_1} \right)
\end{equation}
\end{subequations}
And the dynamic of the system is described by $\dot{F}= \{F,H\} + (F,H;S,H)_{fr(1)} + (F,H;S,H)_{fr(2)} + (F,H;S,H)_{tr}$, with $S = S_1 + S_2$. Inserting for $F$ the functions $p$, $S_1$ and $S_2$ we find 
\begin{equation}
\label{eqn:dynamic_two_piston}
\dot{p}=-\frac{\partial U_1}{\partial x} -\frac{\partial U_2}{\partial x} -(\lambda_1+\lambda_2) \frac{p}{m} \qquad \dot{S_1} = \frac{\lambda_1}{T_1} \frac{p^2}{m^2} + \frac{\kappa}{T_1}(T_2-T_1) \qquad \dot{S_2} = \frac{\lambda_2}{T_2} \frac{p^2}{m^2} + \frac{\kappa}{T_2}(T_1-T_2) ~.
\end{equation}
We indeed recover the good equation of motion as given in \cite{gruber1999thermodynamics} (modulo change of variables).
\section{Euler-Poincaré systems}
\label{sec:EP}
Our theory also applies to Euler-Poincaré systems resulting from the reduction of a Lagrangian on a Lie group with advected parameters (see \cite{holm1998euler} for more details). If this Lagrangian is invariant by the left action of the Lie group, it can be reduced to the Lie algebra giving a variational principle on the Lie algebra. In this case we consider a Lie algebra $\g$ that acts from the left on a vector space $V$ (that will be the space of our non-entropy advected parameters) and $\R$ (for the entropy). The evolution of the system is once again given by a Lagrangian $l : ~ \g \times V \times \R \mapsto \R$ and we suppose that our system is subject to a friction force $F^{fr} \in \g^*$ (since we do not work anymore in $\R^d$, let us be more careful with the spaces in which the objects are defined). We can do the same assumption as before and write $F^{fr} = - \Lambda \dot{q}$, with now $\Lambda$ being a $1$-covariant, $1$-contravariant tensor. The symmetry hypothesis is replaced by supposing that the bilinear form $(a, b) \in \g^2 \mapsto \langle a, \Lambda b \rangle$ is symmetric. Hence it defines a (pseudo-)metric on $\g$.
\subsection{Variational principle}
We consider the following problem:
\begin{problem}[Variational formulation for Euler-Poincaré thermodynamic system]
\label{pbm:variational_form_EP}
Find curves $\xi(t) \in \g, ~ a(t) \in V, ~ s(t) \in \R$ such that
\label{eqn:variational_form_EP}
\begin{equation}
\label{eqn:lagrangian_EP}
\delta \int_0^{T_f} l(\xi, a, s) dt = 0 ~,
\end{equation}
subject to the following variational constraints 
\begin{equation}
\label{eqn:var_constraint_EP}
\delta \xi = \partial_t \eta + [\xi, \eta] ~, \qquad
\delta a + \eta a = 0 ~, \qquad
\frac{\partial l}{\partial s} (\delta S + \eta s)  = \langle F^{fr}, \eta \rangle ~,
\end{equation}
and advection/phenomenological constraints
\begin{equation}
\label{eqn:advection_EP}
\partial_t a + \xi a = 0 ~, \qquad
\frac{\partial l}{\partial s} (\partial_t s + \xi s)  = \langle F^{fr}, \xi \rangle ~,
\end{equation}
for a curve $\eta(t)$ in $\g$ that vanishes at $t=0$ and $t=T_f$.
\end{problem}%
The solution curves to this problem satisfy  
\begin{equation}
\label{eqnEP_mom}
\langle \partial_t \frac{\delta l}{\delta \xi}, \eta \rangle - \langle \frac{\delta l}{\delta \xi},[\xi, \eta]\rangle - \langle \frac{\delta l}{\delta a}, \eta a \rangle - \langle \frac{\delta l}{\delta s}, \eta s \rangle + \langle F^{fr}, \eta \rangle = 0 ~,
\end{equation}
for all $\eta \in \g$. The evolution of entropy is given by \cref{eqn:advection_EP}, which is similar to \cref{eqn:EL_simple_dot_s} with an advection term $\xi s$, and the evolution of the other parameters is purely advection.
\subsection{Metriplectic formulation} 
\begin{definition}
\label{def:mom_Ham_EP}
We introduce the canonical momentum and Hamiltonian:
\begin{equation}
\mu = \frac{\partial l}{\partial \xi} ~,
\end{equation}
\begin{equation}
\label{eqn:ham_EP}
h = \langle \mu, \xi \rangle - l ~.
\end{equation}
\end{definition}
Where compared to the standard momentum, we only take partial derivative with respect to the first variable. We remark that $\frac{\delta h}{\delta \mu} = \xi$, this identity will be useful in the next computation.
In the same fashion as the previous sections, we consider a functional $f(\mu, a, s)$ and compute its time derivative
\begin{align*}
\dot{f} &= \langle \frac{\delta f}{\delta \mu}, \dot{\mu} \rangle + \langle \frac{\delta f}{\delta a}, \dot{a} \rangle + \langle \frac{\delta f}{\delta s}, \dot{s} \rangle \\
	    &= \langle \mu, [\frac{\delta h}{\delta p}, \frac{\delta f}{\delta \mu}] \rangle + \langle \frac{\delta h}{\delta a}, \frac{\delta f}{\delta \mu} a \rangle + \langle \frac{\delta h}{\delta s}, \frac{\delta f}{\delta \mu} a \rangle - \langle F^{fr}, \frac{\delta F}{\delta \mu} \rangle \\
	    &- \langle \frac{\delta f}{\delta a}, \frac{\delta h}{\delta \mu} a \rangle - \langle \frac{\delta f}{\delta s}, \frac{\delta h}{\delta \mu} s \rangle - \frac{\delta f}{\delta s} (\frac{1}{T} \langle F^{fr}, \frac{\delta h}{\delta \mu} \rangle) \\
	    &= \{f,h\} + (f,h;s,h) ~.
\end{align*}
With 
\begin{equation}
\label{eqn:EPLP_bracket}
\{f,h\} = \langle \mu, [\frac{\delta h}{\delta p}, \frac{\delta f}{\delta \mu}] \rangle + \langle \frac{\delta h}{\delta a}, \frac{\delta f}{\delta \mu} a \rangle - \langle \frac{\delta f}{\delta a}, \frac{\delta h}{\delta \mu} a \rangle + \langle \frac{\delta h}{\delta s}, \frac{\delta f}{\delta \mu} a \rangle  - \langle \frac{\delta f}{\delta s}, \frac{\delta h}{\delta \mu} s \rangle ~,
\end{equation}
the standard Lie-Poisson bracket \cite{holm1998euler}, and 
\begin{equation}
\label{eqn:EP_metric_bracket_sym}
(f,g;m,n) = \frac{1}{T}\left(\langle \frac{\partial n}{\partial p},\Lambda \frac{\partial f}{\partial p}\rangle \frac{\partial g}{\partial s}\frac{\partial m}{\partial s} - \langle \frac{\partial n}{\partial p},\Lambda \frac{\partial g}{\partial p} \rangle \frac{\partial f}{\partial s}\frac{\partial m}{\partial s}-\langle \frac{\partial m}{\partial p},\Lambda \frac{\partial f}{\partial p}\rangle \frac{\partial g}{\partial s}\frac{\partial n}{\partial s} + \langle \frac{\partial m}{\partial p},\Lambda \frac{\partial g}{\partial p} \rangle \frac{\partial f}{\partial s}\frac{\partial n}{\partial s}\right) ~,
\end{equation}
the same dissipative bracket as in the non-reduced case, where $T = \frac{\delta h}{\delta s}$. Hence, we have proved the following theorem:
\begin{theorem}
A curve $(\xi(t), a(t),s(t)) \in \g \times V \times \R$ is solution to \cref{pbm:variational_form_EP} if and only if, doing the change of variables $(\xi, a, s) \mapsto (\mu, a, s)$, it is a metriplectic system \ref{def:metriplectic} with the brackets defined by \cref{eqn:EPLP_bracket,eqn:EP_metric_bracket_sym}, and the Hamiltonian \cref{eqn:ham_EP}.
\end{theorem}
\subsection{Example: Viscous and heat conducting fluids}
We now show how to apply the previous result to a more involved example: the thermodynamic Navier-Stokes equations. For the sake of simplicity and in the goal of illustrating the previous result, we only present here the principle in spatial representation. This model is expressed in the variables $\uu$ (the velocity of the fluid), $\rho$ (the density of the fluid) and $s$ (its entropy). We suppose that the fluid is evolving in a domain $\Omega \subset \R^3$, that the system is closed and there is no exterior power supply. The Lagrangian is 
\begin{equation}
\label{eqn:lag_visc_fluid}
l(\uu, \rho, \sigma) = \int_\Omega \frac{1}{2}\rho |\uu^2| - e(\rho, s) d\xx ~,
\end{equation}
where $e$ is the internal energy of the fluid (as a function of $\rho$ and $s$). We denote $X(\Omega)$ the set of vector fields on $\Omega$ (that is smooth functions $\Omega \rightarrow \R^3$) and $F(\Omega)$ the set of real valued (smooth) functions on $\Omega$.
Solutions to the thermodynamic Navier-Stokes equations can also be found as solutions of the following problem:
\begin{problem}[Variational Formulation for thermodynamic Navier-Stokes]
\label{pbm:var_NS}
Find curves $s(t), \sigma(t), \rho(t), \gamma(t) \in F(\Omega)$ and $\uu(t) \in X(\Omega)$ that extremize the following functional: 
\begin{equation}
\label{eqn:var_princip_fluid}
\int_{0}^{T_f} (l(\uu, \rho, s) + \int_\Omega (s-\sigma) (\partial_t \gamma + \uu \cdot \nabla \gamma)) dt ~,
\end{equation}
under the constraints 
\begin{align}
\label{eqn:constraints_fluid}
\delta \uu = \partial_t \vv + [\uu,\vv]~, \qquad \delta \rho = -\nabla \cdot(\rho\vv)~, \qquad \partial_t \rho + \nabla \cdot (\rho \uu) = 0 ~, \\
\frac{\delta l}{\delta s}(\partial_t \sigma + \nabla \cdot (\sigma \uu)) = - \ssig^{fr} : \nabla \uu + j_s \cdot \nabla (\partial_t \gamma + \uu \cdot \nabla \gamma) ~, \\
\frac{\delta l}{\delta s}(\delta \sigma + \nabla \cdot (\sigma \vv)) = - \ssig^{fr} : \nabla \vv + j_s \cdot \nabla (\delta \gamma + \vv \cdot \nabla \gamma ) ~.
\end{align} 
where $\vv(t)$ is a curve on $X(\Omega)$ (a time dependent vector field) that vanishes at $t=0$ and $t=T$.
\end{problem}
The first line of constraints comes purely from the advection of quantities, whereas the second and third lines come from both advection of parameters and thermodynamic constraints. $\ssig^{fr}$ is the friction tensor, which we suppose to have the form 
\begin{equation}
\label{eqn:structure_viscosity}
\ssig^{fr} = \Lambda \nabla \uu ~,
\end{equation}
where $\Lambda$ is a 4-tensor that has the symmetry $\Lambda_{ijkl} = \Lambda_{klij}$ and is positive in the sense that $a_{ij}\Lambda_{ijkl}a_{kl}\geq 0$ for any two tensor $(a_{ij})$. Note that since we work in $\R^3$ we do not distinguish between covariant and contravariant indices. We also have introduced $j_s$ the heat flux and we suppose according to Fourier law
\begin{equation}
\label{eqn:Fourier}
T j_s = - \kappa \nabla T ~.
\end{equation}
After computation (see for example \cite{gawlik2022variational,gay2017lagrangian2}) we see that solutions to this variational problem satisfy
\begin{align}
\label{eqn:variation_v}
\int_\Omega \partial_t (\rho \uu) \cdot \vv &= \int_\Omega \rho \uu \cdot [\uu, \vv] + \frac{\delta e}{\delta \rho} \nabla \cdot (\rho \vv) + \frac{\delta e}{\delta s} \nabla \cdot (s \vv) - \ssig^{fr} : \nabla \vv ~,
\end{align}
for any (smooth) vector field $\vv$ and at any time $t$. We also have the evolution equations of density and entropy 
\begin{subequations}
\label{eqn:density_entropy_fluid}
\begin{equation}
\partial_t \rho + \nabla \cdot (\rho \uu) = 0 ~,
\end{equation}
\begin{equation}
\partial_t s + \nabla \cdot (s \uu) = \frac{1}{T}(\ssig^{fr} : \nabla \uu - j_s \cdot \nabla \frac{\delta l}{\delta s}) - \nabla \cdot(j_s) ~.
\end{equation}
\end{subequations}
We introduce the momentum and hamiltonian:
\begin{subequations}
\label{eqn:mom_Ham_fluid}
\begin{equation}
\mm = \frac{\delta l}{\delta \uu} = \rho \uu ~,
\end{equation}
\begin{equation}
h = \langle \mm, \uu \rangle - l = \int_\Omega \frac{\mm^2}{2 \rho} + e(\rho, s) d\xx ~.
\end{equation}
\end{subequations}
We start by denoting that $\frac{\delta h}{\delta \mm} = \uu$, $\frac{\delta h}{\delta \rho} = \frac{\delta e}{\delta \rho}$ and $\frac{\delta h}{\delta s} = \frac{\delta e}{\delta s}$.
Let $\mm$, $\rho$ and $s$ being solution of this equation, and consider a functional $f(\mm, \rho, s)$, we will do the same computation as before to rewrite those equations in a metriplectic setting.
\begin{align*}
\dot{f} &= \int_\Omega \frac{\delta f}{\delta \mm} \dot{\mm} + \frac{\delta f}{\delta \rho} \dot{\rho} + \frac{\delta f}{\delta s} \dot{s} \\
		&= \int_\Omega \mm \cdot [\uu, \frac{\delta f}{\delta \mm}] + \frac{\delta e}{\delta \rho} \nabla \cdot (\rho \frac{\delta f}{\delta \mm}) + \frac{\delta e}{\delta s} \nabla \cdot (s \frac{\delta f}{\delta \mm}) - \ssig^{fr} : \nabla \frac{\delta f}{\delta \mm} - \frac{\delta f}{\delta \rho} \nabla \cdot (\rho \uu) \\
		& - \frac{\delta f}{\delta s} \nabla \cdot (s \uu) + \frac{\delta f}{\delta s}\frac{1}{T}(\ssig^{fr} : \nabla \uu -j_s \cdot \nabla \frac{\delta h}{\delta s}) - \frac{\delta f}{\delta s} \nabla \cdot(j_s) \\
		&= \int_\Omega \mm \cdot [\frac{\delta h}{\delta \mm}, \frac{\delta f}{\delta \mm}] + \frac{\delta h}{\delta \rho} \nabla \cdot (\rho \frac{\delta f}{\delta \mm}) - \frac{\delta f}{\delta \rho} \nabla \cdot (\rho \frac{\delta h}{\delta \mm}) + \frac{\delta h}{\delta s} \nabla \cdot (s \frac{\delta f}{\delta \mm}) - \frac{\delta f}{\delta s} \nabla \cdot (s \frac{\delta h}{\delta \mm})  \\
		& + \frac{1}{T} \left(\frac{\delta f}{\delta s} (\nabla\frac{\delta h}{\delta \mm} : \Lambda : \nabla\frac{\delta h}{\delta \mm}) - \frac{\delta h}{\delta s} (\nabla\frac{\delta h}{\delta \mm} : \Lambda : \nabla\frac{\delta f}{\delta \mm}  + \frac{\kappa}{T}\left(\frac{\delta f}{\delta s} \nabla \frac{\delta h}{\delta s} \cdot \nabla \frac{\delta h}{\delta s} - \frac{\delta h}{\delta s} \nabla \frac{\delta f}{\delta s} \cdot \nabla \frac{\delta h}{\delta s}\right)\right) \\
		&= \{f,h\} + \int_\Omega \frac{1}{T}\left(
		\begin{array}{cr}
		\nabla\frac{\delta h}{\delta \mm} : \Lambda : \nabla\frac{\delta h}{\delta \mm}) \frac{\delta f}{\delta s} \frac{\delta s}{\delta s}  
		- (\nabla\frac{\delta h}{\delta \mm} : \Lambda : \nabla\frac{\delta f}{\delta \mm})\frac{\delta h}{\delta s} \frac{\delta s}{\delta s} 
		+ (\nabla\frac{\delta s}{\delta \mm} : \Lambda : \nabla\frac{\delta f}{\delta \mm}) \frac{\delta h}{\delta s} \frac{\delta h}{\delta s} 
		- (\nabla\frac{\delta s}{\delta \mm} : \Lambda : \nabla\frac{\delta h}{\delta \mm})\frac{\delta f}{\delta s} \frac{\delta h}{\delta s}) 
		\end{array}
		\right) \\
		&+ \int_\Omega \frac{\kappa}{T^2}\left(\nabla \frac{\delta h}{\delta s} \cdot \nabla \frac{\delta h} {\delta s} \frac{\delta f}{\delta s} \frac{\delta s}{\delta s} - \nabla \frac{\delta f}{\delta s} \cdot \nabla \frac{\delta h}{\delta s} \frac{\delta h}{\delta s} \frac{\delta s}{\delta s} 
        + \nabla \frac{\delta s}{\delta s} \cdot \nabla \frac{\delta f} {\delta s} \frac{\delta h}{\delta s} \frac{\delta h}{\delta s} - \nabla \frac{\delta s}{\delta s} \cdot \nabla \frac{\delta h}{\delta s} \frac{\delta f}{\delta s} \frac{\delta h}{\delta s}\right)\\
        &= \{f,h\} + (f,h;s,h)_{visc} + (f,h;s,h)_{heat} ~ ,
\end{align*}
with 
\begin{subequations}
\label{eqn:brackets_fluid}
    \begin{align}
    \label{eqn:Lie_poisson_fluid}
        \{f,g\} &= \int_\Omega \mm \cdot [\frac{\delta g}{\delta \mm}, \frac{\delta f}{\delta \mm}] + \frac{\delta g}{\delta \rho} \nabla \cdot (\rho \frac{\delta f}{\delta \mm}) - \frac{\delta f}{\delta \rho} \nabla \cdot (\rho \frac{\delta g}{\delta \mm}) + \frac{\delta g}{\delta s} \nabla \cdot (s \frac{\delta f}{\delta \mm}) - \frac{\delta f}{\delta s} \nabla \cdot (s \frac{\delta g}{\delta \mm}) ~, \\
        (f,g;m,n)_{visc} &= \int_\Omega \frac{1}{T}\left(
        \begin{array}{cr}
        \nabla\frac{\delta n}{\delta \mm} : \Lambda : \nabla\frac{\delta g}{\delta \mm}) \frac{\delta f}{\delta s} \frac{\delta m}{\delta s}  
        - (\nabla\frac{\delta n}{\delta \mm} : \Lambda : \nabla\frac{\delta f}{\delta \mm})\frac{\delta g}{\delta s} \frac{\delta m}{\delta s} \\
		+ (\nabla\frac{\delta m}{\delta \mm} : \Lambda : \nabla\frac{\delta f}{\delta \mm}) \frac{\delta g}{\delta s} \frac{\delta n}{\delta s}  
		- (\nabla\frac{\delta m}{\delta \mm} : \Lambda : \nabla\frac{\delta g}{\delta \mm})\frac{\delta f}{\delta s} \frac{\delta n}{\delta s}) 
        \end{array} 
        \right) ~,\\
        (f,g;m,n)_{heat} &= \int_\Omega \frac{\kappa}{T^2}\left(\nabla \frac{\delta n}{\delta s} \cdot \nabla \frac{\delta g}{\delta s} \frac{\delta f}{\delta s} \frac{\delta m}{\delta s} - \nabla \frac{\delta n}{\delta s} \cdot \nabla \frac{\delta f}{\delta s} \frac{\delta g}{\delta s} \frac{\delta m}{\delta s} 
        +\nabla \frac{\delta m}{\delta s} \cdot \nabla \frac{\delta f}{\delta s} \frac{\delta g}{\delta s} \frac{\delta n}{\delta s} - \nabla \frac{\delta m}{\delta s} \cdot \nabla \frac{\delta g}{\delta s} \frac{\delta f}{\delta s} \frac{\delta n}{\delta s}\right) ~.
    \end{align}
\end{subequations}
The symplectic part is exactly the standard Lie-Poisson bracket \cite{holm1998euler}, while the dissipative parts are now rewritten with the metriplectic 4-bracket formalism. We recover the brackets proposed in \cite{morrison2024inclusive} in a more constructive way. We summarize this results in
\begin{theorem}
A curve $(\rho(t),s(t),\sigma(t),\gamma(t),\uu(t) \in F(\Omega)^4 \times X(\Omega))$ is solution to \cref{pbm:var_NS} if and only if, doing the change of variables $(\rho(t),s(t), \uu(t)) \mapsto (\rho(t),s(t), \mm(t))$, it is a metriplectic system \ref{def:metriplectic} with the brackets defined by \cref{eqn:brackets_fluid}, the Hamiltonian as in \cref{eqn:mom_Ham_fluid}, and $\sigma$ satisfies $\dot{\sigma} = \dot{s} + \nabla \cdot((s-\sigma)\uu) + \nabla \cdot j_s$ and $\dot{\gamma} + \uu\cdot \nabla \gamma = T$.
\end{theorem}
We have added the equations to recover $\sigma$ and $\gamma$ for completeness, but in practice, one is generally only interested in $\rho, s$ and $\uu$.

We can reduce these 4-brackets to 2-brackets as done previously:
\begin{align*}
(f,g)_{visc} &= (f,h;g,h)_{visc} \\
	&=\int_\Omega \frac{1}{T}\Big(
        \nabla\frac{\delta h}{\delta \mm} : \Lambda : \nabla\frac{\delta h}{\delta \mm}) \frac{\delta f}{\delta s} \frac{\delta g}{\delta s}  
        - (\nabla\frac{\delta h}{\delta \mm} : \Lambda : \nabla\frac{\delta f}{\delta \mm})\frac{\delta h}{\delta s} \frac{\delta g}{\delta s} \\
		& + (\nabla\frac{\delta g}{\delta \mm} : \Lambda : \nabla\frac{\delta f}{\delta \mm}) \frac{\delta h}{\delta s} \frac{\delta h}{\delta s}  
		- (\nabla\frac{\delta g}{\delta \mm} : \Lambda : \nabla\frac{\delta h}{\delta \mm})\frac{\delta f}{\delta s} \frac{\delta h}{\delta s}) 
        \Big) \\
   &=\int_\Omega \frac{1}{T}\Big(
        \nabla\frac{\delta h}{\delta \mm} : \Lambda : \nabla\frac{\delta h}{\delta \mm}) \frac{\delta f}{\delta s} \frac{\delta g}{\delta s}  
        - (\nabla\frac{\delta h}{\delta \mm} : \Lambda : \nabla\frac{\delta f}{\delta \mm}) T \frac{\delta g}{\delta s} \\
		&+ (\nabla\frac{\delta g}{\delta \mm} : \Lambda : \nabla\frac{\delta f}{\delta \mm}) T^2  
		- (\nabla\frac{\delta g}{\delta \mm} : \Lambda : \nabla\frac{\delta h}{\delta \mm})T \frac{\delta f}{\delta s}) 
		\Big) \\
	&= \int_\Omega T \Big( \nabla \frac{\delta f}{\delta \mm} - \frac{1}{T}\nabla \frac{\delta h}{\delta \mm} \frac{\delta f}{\delta s} \Big) : \Lambda :
	\Big( \nabla \frac{\delta g}{\delta \mm} - \frac{1}{T}\nabla \frac{\delta h}{\delta \mm} \frac{\delta g}{\delta s} \Big)        
\end{align*}

\begin{align*}
(f,g)_{heat} &= (f,h;g,h)_{heat} \\
	&= \int_\Omega \frac{\kappa}{T^2}\left(\nabla \frac{\delta h}{\delta s} \cdot \nabla \frac{\delta h}{\delta s} \frac{\delta f}{\delta s} \frac{\delta g}{\delta s} - \nabla \frac{\delta h}{\delta s} \cdot \nabla \frac{\delta f}{\delta s} \frac{\delta h}{\delta s} \frac{\delta g}{\delta s} 
        +\nabla \frac{\delta g}{\delta s} \cdot \nabla \frac{\delta f}{\delta s} \frac{\delta h}{\delta s} \frac{\delta h}{\delta s} - \nabla \frac{\delta g}{\delta s} \cdot \nabla \frac{\delta h}{\delta s} \frac{\delta f}{\delta s} \frac{\delta h}{\delta s}\right) \\
        &= \int_\Omega \frac{\kappa}{T^2}\left(\nabla T\cdot \nabla T \frac{\delta f}{\delta s} \frac{\delta g}{\delta s} - \nabla T \cdot \nabla \frac{\delta f}{\delta s} T \frac{\delta g}{\delta s} 
        +\nabla \frac{\delta g}{\delta s} \cdot \nabla \frac{\delta f}{\delta s} T^2 - \nabla \frac{\delta g}{\delta s} \cdot \nabla T \frac{\delta f}{\delta s} T\right)  \\
        &= \int_\Omega \kappa T^2\left(\frac{\nabla T}{T^2}\cdot \frac{\nabla T}{T^2} \frac{\delta f}{\delta s} \frac{\delta g}{\delta s} - \frac{\nabla T}{T^2} \cdot \nabla \frac{\delta f}{\delta s} \frac{1}{T} \frac{\delta g}{\delta s} 
        +\nabla \frac{\delta g}{\delta s} \cdot \nabla \frac{\delta f}{\delta s} \frac{1}{T^2} - \nabla \frac{\delta g}{\delta s} \cdot \frac{\nabla T}{T^2} \frac{\delta f}{\delta s} \frac{1}{T}\right) \\
        &= \int_\Omega \kappa T^2\left(\nabla \frac{1}{T}\cdot \nabla \frac{1}{T} \frac{\delta f}{\delta s} \frac{\delta g}{\delta s} + \nabla \frac{1}{T} \cdot \nabla \frac{\delta f}{\delta s} \frac{1}{T} \frac{\delta g}{\delta s} 
        +\nabla \frac{\delta g}{\delta s} \cdot \nabla \frac{\delta f}{\delta s} \frac{1}{T^2} + \nabla \frac{\delta g}{\delta s} \cdot \nabla \frac{1}{T} \frac{\delta f}{\delta s} \frac{1}{T}\right) \\
        &= \int_\Omega \kappa T^2 \left(\nabla \frac{1}{T} \frac{\delta f}{\delta s} + \frac{1}{T} \nabla \frac{\delta f}{\delta s} \right) \left(\nabla \frac{1}{T} \frac{\delta g}{\delta s} + \frac{1}{T} \nabla \frac{\delta g}{\delta s} \right)  \\
        &= \int_\Omega \kappa T^2 \nabla \left(\frac{1}{T} \frac{\delta f}{\delta s} \right) \nabla \left( \frac{1}{T} \frac{\delta g}{\delta s} \right) \\
\end{align*}

We recover the metriplectic brackets described in \cite{eldred2020single}.

\section{Systems with no symplectic part}
\label{sec:no_symp}
As mentioned in \cref{sec:simple}, all our previous derivations are not valid when dealing with systems where the Legendre transform $\dot{q} \mapsto p$ is not invertible. We now derive a metriplectic formulation for such systems, where there is no symplectic part. Consider a Lagrangian $L = L(q,S)$, if $S$ is assumed a simple parameter of the model, then the extremality condition reads $\delta L(q) =0$ and there is no dynamic, the system is steady at an equilibrium. However the variational formulation proposed in \cite{gay2017lagrangian1} allows to understand the dynamic of such system imposing a variational condition on the evolution of entropy.
\subsection{Equation of motion}
We consider the following variational problem:
\begin{problem}[Variational formulation for system with no symplectic part]
\label{pbm:variational_form_no_symp}
Find curve $(q(t), S(t)) \in Q \times \R$ that satisfies the variational condition
\begin{equation}
\label{eqn:lagrangian_no_symp}
\delta \int_0^T L(q, S) dt = 0 ~,
\end{equation}
with variation subject to the variational constraint
\begin{equation}
\label{eqn:var_constraint_no_symp}
\frac{\partial L}{\partial S} \delta S = \langle F^{fr}(q,S), \delta q \rangle ~,
\end{equation}
and the phenomenological constraint
\begin{equation}
\label{eqn:pheno_constrain_no_symp}
\frac{\partial L}{\partial S} \dot{S} = \langle F^{fr}(q,S), \dot{q} \rangle ~.
\end{equation}
\end{problem}
We keep assumption \cref{eqn:structure_dissipation} on the dissipation, adding that $\Lambda$ should be invertible, we denote $\Gamma$ its inverse. 
\begin{remark}
    Since $\Lambda$ is symmetric, we can diagonalize it. In the directions with $0$ eigenvalues (let name one of them $q_0$), the variational problem gives $\frac{\delta L}{\delta q_0} = 0$ and no dynamic, therefore we can remove this direction from the problem and work with an invertible matrix.
\end{remark}
The equations given by the variational problem are
\begin{subequations}
    \label{eqn:no_symp_varcond}
    \begin{equation}
        \frac{\delta L}{\delta S} \dot{S} = -\dot{q} \Lambda \dot{q} ~,
    \end{equation}
    \begin{equation}
        \frac{\delta L}{\delta q} = \Lambda \dot{q} ~.
    \end{equation}
\end{subequations}
\subsection{Metriplectic formulation}
Considering the Hamiltonian $H = -L$, we rewrite the equation of motion 
\begin{subequations}
    \label{eqn:EL_no_symp_varcond}
    \begin{equation}
        \dot{S} = \frac{1}{T} \frac{\delta H}{\delta q} \Gamma \frac{\delta H}{\delta q} ~,
    \end{equation}
    \begin{equation}
        \dot{q}  = -\Gamma \frac{\delta H}{\delta q} ~.
    \end{equation}
\end{subequations}
Consider a functional $F(q,S)$ :
\begin{align*}
    \dot{F} &= \frac{\delta F}{\delta q} \dot{q} + \frac{\delta F}{\delta S} \dot{S} \\
    &= \frac{1}{T} \frac{\delta H}{\delta q} \Gamma \frac{\delta H}{\delta q} \frac{\delta F}{\delta S} - \frac{\delta F}{\delta q} \Gamma \frac{\delta H}{\delta q} \\
    &= \frac{1}{T} ( \frac{\delta H}{\delta q} \Gamma \frac{\delta H}{\delta q} \frac{\delta F}{\delta S} \frac{\delta S}{\delta S} - \frac{\delta F}{\delta q} \Gamma \frac{\delta H}{\delta q} \frac{\delta H}{\delta S} \frac{\delta S}{\delta S} + \frac{\delta S}{\delta q} \Gamma \frac{\delta F}{\delta q} \frac{\delta H}{\delta S} \frac{\delta H}{\delta S} - \frac{\delta S}{\delta q} \Gamma \frac{\delta H}{\delta q} \frac{\delta F}{\delta S} \frac{\delta H}{\delta S}) \\
    &= (F,H;S,H) ~,
\end{align*}
with 
\begin{equation}
    \label{eqn:metri_bracket_no_symp}
    (F,G;M,N) = \frac{1}{T} ( \frac{\delta N}{\delta q} \Gamma \frac{\delta G}{\delta q} \frac{\delta F}{\delta S} \frac{\delta M}{\delta S} - \frac{\delta F}{\delta q} \Gamma \frac{\delta N}{\delta q} \frac{\delta G}{\delta S} \frac{\delta M}{\delta S} + \frac{\delta M}{\delta q} \Gamma \frac{\delta F}{\delta q} \frac{\delta G}{\delta S} \frac{\delta N}{\delta S} - \frac{\delta M}{\delta q} \Gamma \frac{\delta G}{\delta q} \frac{\delta F}{\delta S} \frac{\delta N}{\delta S}) ~,
\end{equation}
therefore this equation can be seen as a pure dissipative metriplectic motion.
\begin{theorem}
A curve $(q(t), S(t)) \in Q \times \R$ is solution to \cref{pbm:variational_form_no_symp} if and only if it is a metriplectic system \ref{def:metriplectic} with the bracket defined by \cref{eqn:metri_bracket_no_symp}, no symplectic bracket and the Hamiltonian $H = -L$.
\end{theorem}
\subsection{Example: Simple system with chemical reactions}
We consider a system of several components undergoing $r$ chemical reactions. The state of the system is given by the generalized coordinates $(\psi_i)_{i=1,...,r}$, representing the degree of advancement of each reaction. We introduce the internal energy of the system $U(\psi_1,...,\psi_r,S)$ as a function of the degrees of advancement and the entropy $S$. The Lagrangian of the system is 
\begin{equation}
\label{eqn:lag_chemical}
L(\psi_i,S) = -U(\psi_i,S) ~.
\end{equation}
We suppose that each reaction is subject to a dissipation, related to the advancement of the other reactions 
\begin{equation}
\label{eqn:friction_chemical}
F^{fr(i)} = - \sum\limits_{j=1}^r \lambda_{ij} \dot{\psi}_j ~,
\end{equation}
where the matrix $(\lambda_{ij})$ is symmetric definite positive. As previously stated we can now define our Hamiltonian
\begin{equation}
\label{eqn:Ham_bracket_chemical}
H = U(\psi_i,S) ~.
\end{equation}
And the bracket defined by \cref{eqn:metri_bracket_no_symp}. Considering the dynamics $\dot{F} = (F,H;S,H)$ we recover the equations:
\begin{equation}
\label{eqn:dynamic_chemical}
\dot{S} = \frac{1}{T}(\frac{\partial U}{\partial \psi} \Gamma \frac{\partial U}{\partial \psi}) \qquad \dot{q_i} = -\sum\limits_j \Gamma_{ij} \frac{\partial U}{\partial \psi_j} ~,
\end{equation}
giving indeed the evolution of a chemical system as described in \cite{gay2017lagrangian1}.

\section{Conclusion and perspective}
\label{sec:concl}
In this work we presented a metriplectic reformulation of the variational formulation for non-equilibrium thermodynamics. This reformulation has been done by studying the evolution of observable of the system and rewriting it using canonical brackets and metriplectic 4-brackets. We observed that the latter has particular Kulkarni-Nomizu product structure which might provide some geometrical interpretation of this principle. This formulation is valid for a broad class of systems and is extendable to more intricate system as has being shown with the thermodynamical Navier-Stokes equation. 

Future work will focus on applying this framework to new systems in order to provide a metriplectic description of more complicated settings, and using this equivalence to construct new structure preserving algorithms.

\section*{Acknowledgements}
The author thanks Martin Campos-Pinto and Eric Sonnendrücker for their help during the redaction of this manuscript. Inspiring discussion with Philip Morrison, William Barham, Michael Kraus and Omar Maj are also warmly acknowledged. The author is also thankful to Francois Gay-Balmaz for his feedback during the redaction of this work.
\bibliographystyle{plain} 
\bibliography{var_MHD}

\begin{thebibliography}{10}

\bibitem{biot1975virtual}
Maurice~Anthony Biot.
\newblock A virtual dissipation principle and lagrangian equations in
  non-linear irreversible thermodynamics.
\newblock {\em Bulletins de l'Acad{\'e}mie Royale de Belgique}, 61(1):6--30,
  1975.

\bibitem{carlier2025variational}
Valentin Carlier and Martin Campos~Pinto.
\newblock Variational discretizations of ideal magnetohydrodynamics in smooth
  regime using structure-preserving finite elements.
\newblock {\em Journal of Computational Physics}, 523:113647, February 2025.

\bibitem{eldred2020single}
Christopher Eldred and Fran{\c{c}}ois Gay-Balmaz.
\newblock Single and double generator bracket formulations of multicomponent
  fluids with irreversible processes.
\newblock {\em Journal of Physics A: Mathematical and Theoretical},
  53(39):395701, 2020.

\bibitem{eldred2021thermodynamically}
Christopher Eldred and Fran{\c{c}}ois Gay-Balmaz.
\newblock Thermodynamically consistent semi-compressible fluids: a variational
  perspective.
\newblock {\em Journal of Physics A: Mathematical and Theoretical},
  54(34):345701, 2021.

\bibitem{gawlik2021variational}
Evan~S Gawlik and Fran{\c{c}}ois Gay-Balmaz.
\newblock A variational finite element discretization of compressible flow.
\newblock {\em Foundations of Computational Mathematics}, 21:961--1001, 2021.

\bibitem{gawlik2022variational}
Evan~S Gawlik and Fran{\c{c}}ois Gay-Balmaz.
\newblock Variational and thermodynamically consistent finite element
  discretization for heat conducting viscous fluids.
\newblock {\em arXiv preprint arXiv:2211.08745}, 2022.

\bibitem{gawlik2011geometric}
Evan~S Gawlik, Patrick Mullen, Dmitry Pavlov, Jerrold~E Marsden, and Mathieu
  Desbrun.
\newblock Geometric, variational discretization of continuum theories.
\newblock {\em Physica D: Nonlinear Phenomena}, 240(21):1724--1760, 2011.

\bibitem{gay2017lagrangian1}
Fran{\c{c}}ois Gay-Balmaz and Hiroaki Yoshimura.
\newblock A {Lagrangian} variational formulation for nonequilibrium
  thermodynamics. {Part I}: discrete systems.
\newblock {\em Journal of Geometry and Physics}, 111:169--193, 2017.

\bibitem{gay2017lagrangian2}
Fran{\c{c}}ois Gay-Balmaz and Hiroaki Yoshimura.
\newblock A {Lagrangian} variational formulation for nonequilibrium
  thermodynamics. {Part II}: continuum systems.
\newblock {\em Journal of Geometry and Physics}, 111:194--212, 2017.

\bibitem{gay2020variational}
Fran{\c{c}}ois Gay-Balmaz and Hiroaki Yoshimura.
\newblock From variational to bracket formulations in nonequilibrium
  thermodynamics of simple systems.
\newblock {\em Journal of Geometry and Physics}, 158:103812, 2020.

\bibitem{grmela1997dynamics}
Miroslav Grmela and Hans~Christian {\"O}ttinger.
\newblock Dynamics and thermodynamics of complex fluids. i. development of a
  general formalism.
\newblock {\em Physical Review E}, 56(6):6620, 1997.

\bibitem{gruber1999thermodynamics}
Christian Gruber.
\newblock Thermodynamics of systems with internal adiabatic constraints: time
  evolution of the adiabatic piston.
\newblock {\em European journal of physics}, 20(4):259, 1999.

\bibitem{gyarmati1970non}
Istvan Gyarmati et~al.
\newblock {\em Non-equilibrium thermodynamics}, volume 184.
\newblock Springer, 1970.

\bibitem{holm1998euler}
Darryl~D Holm, Jerrold~E Marsden, and Tudor~S Ratiu.
\newblock The {E}uler--{P}oincar{\'e} equations and semidirect products with
  applications to continuum theories.
\newblock {\em Advances in Mathematics}, 137(1):1--81, 1998.

\bibitem{kaufman1984dissipative}
Allan~N Kaufman.
\newblock Dissipative hamiltonian systems: A unifying principle.
\newblock {\em Physics Letters A}, 100(8):419--422, 1984.

\bibitem{kraus2017metriplectic}
Michael Kraus and Eero Hirvijoki.
\newblock Metriplectic integrators for the {Landau} collision operator.
\newblock {\em Physics of Plasmas}, 24(10), 2017.

\bibitem{kulkarni1972bianchi}
Ravindra~S Kulkarni.
\newblock On the bianchi identities.
\newblock {\em Mathematische Annalen}, 199(4):175--204, 1972.

\bibitem{marsden2013introduction}
Jerrold~E Marsden and Tudor~S Ratiu.
\newblock {\em Introduction to mechanics and symmetry: a basic exposition of
  classical mechanical systems}, volume~17.
\newblock Springer Science \& Business Media, 2013.

\bibitem{morrison1984bracket}
Philip~J Morrison.
\newblock Bracket formulation for irreversible classical fields.
\newblock {\em Physics Letters A}, 100(8):423--427, 1984.

\bibitem{morrison1986paradigm}
Philip~J Morrison.
\newblock A paradigm for joined {Hamiltonian} and dissipative systems.
\newblock {\em Physica D: Nonlinear Phenomena}, 18(1-3):410--419, 1986.

\bibitem{morrison2024inclusive}
Philip~J Morrison and Michael~H Updike.
\newblock Inclusive curvaturelike framework for describing dissipation:
  Metriplectic 4-bracket dynamics.
\newblock {\em Physical Review E}, 109(4):045202, 2024.

\bibitem{onsager1931reciprocal}
Lars Onsager.
\newblock Reciprocal relations in irreversible processes. i.
\newblock {\em Physical review}, 37(4):405, 1931.

\bibitem{pavlov2011structure}
Dmitry Pavlov, Patrick Mullen, Yiying Tong, Eva Kanso, Jerrold~E Marsden, and
  Mathieu Desbrun.
\newblock Structure-preserving discretization of incompressible fluids.
\newblock {\em Physica D: Nonlinear Phenomena}, 240(6):443--458, 2011.

\bibitem{prigogine1947etude}
Ilya Prigogine.
\newblock Etude thermodynamique des ph{\'e}nom{\`e}nes irr{\'e}versibles.
\newblock {\em These d'agregation presentee a la taculte des sciences de
  I'Universite Libre de Bruxelles 1945}, 1947.

\bibitem{sanz1992symplectic}
Jesus~M Sanz-Serna.
\newblock Symplectic integrators for {Hamiltonian} problems: an overview.
\newblock {\em Acta numerica}, 1:243--286, 1992.

\bibitem{stueckelberg1974thermocinetique}
ECG Stueckelberg and PB~Scheurer.
\newblock thermocin{\'e}tique ph{\'e}nom{\'e}nologique galil{\'e}enne birkauser
  verlag.
\newblock {\em Basel and Stuttgart}, 1974.

\end{thebibliography}

\end{document}